\documentclass[10pt, twocolumn, cls]{IEEEtran}
\usepackage{amsfonts}
\usepackage[tbtags]{amsmath}
\usepackage{graphicx,subfigure}
\usepackage{color}      
\usepackage{amssymb}
\usepackage{tipa}
\usepackage{bbm}
\usepackage{threeparttable}
\long\def\comment#1{}

\newcommand{\beq}{\begin{equation}}
\newcommand{\eeq}{\end{equation}}
\newcommand{\beqno}{\begin{equation*}}
\newcommand{\eeqno}{\end{equation*}}

\newcommand{\bes}{\begin{split}}
\newcommand{\ees}{\end{split}}
\newcommand{\bdm}{\begin{displaymath}}
\newcommand{\edm}{\end{displaymath}}

\newtheorem{lemma}{Lemma}

\newtheorem{definition}{Definition}

\newcommand{\bd}{\begin{definition}}
\newcommand{\ed}{\end{definition}}

\newcommand{\bv}{\begin{vugraph}}
\newcommand{\ev}{\end{vugraph}}
\newcommand{\bi}{\begin{itemize}}
\newcommand{\ei}{\end{itemize}}
\newcommand{\ben}{\begin{enumerate}}
\newcommand{\een}{\end{enumerate}}

\newcommand{\bean}{\begin{eqnarray*} }
\newcommand{\eean}{\end{eqnarray*} }
\newcommand{\bea}{\begin{eqnarray} }
\newcommand{\eea}{\end{eqnarray} }

\newcommand{\ba}{\begin{array} }
\newcommand{\ea}{\end{array} }

\begin{document}

\title{Improved Combinatorial Algorithms for Wireless Information Flow}

\author{\authorblockN{Cuizhu Shi and Aditya Ramamoorthy}\\
\authorblockA{Department of Electrical and Computer Engineering, Iowa State University, Ames, Iowa 50011\\
Email: \{cshi, adityar\}@iastate.edu}
\thanks{This work was supported in part by NSF grant CNS-0721453 and NSF grant CCF-1018148.}}

\maketitle

\begin{abstract}

The work of Avestimehr et al. `07 has recently proposed a deterministic model for wireless networks and characterized the unicast capacity $C$ of such networks as the minimum rank of the adjacency matrices describing all possible source-destination cuts.
Amaudruz \& Fragouli first proposed a polynomial-time algorithm for finding the unicast capacity of a linear deterministic wireless network in their 2009 paper.
In this work, we improve upon Amaudruz \& Fragouli's work and further reduce the computational complexity of the algorithm by fully exploring the useful combinatorial features intrinsic in the problem. Our improvement applies generally with any size of finite fields associated with the channel model. Comparing with other algorithms on solving the same problem, our improved algorithm is very competitive in terms of complexity.

\end{abstract}

\section{Introduction\label{introduction}}

The deterministic channel model for wireless networks proposed by Avestimehr, Diggavi and Tse \cite{amir2007_deterministicmodel} \cite{amir2007_wirelessnetworkinfoflow} (referred to as ADT model thereafter) has been a useful tool for understanding the fundamental limitations of information transfer in wireless networks.
The ADT model captures two main features, the broadcasting and interference, that are present in wireless networks. It converts the wireless networks into deterministic networks, by making appropriate assumptions, that in turn lead to approximate capacity results.

Consider a point-to-point Gaussian channel given by $y=\sqrt{{\mbox{SNR}}}x+z$ where $z\sim{\cal{N}}(0,1)$ (${\cal{N}}$ represents Gaussian distribution). Assume $x$ and $z$ are real numbers, then we can write $y\approx 2^n\sum_{i=1}^nx(i)2^{-i}+\sum_{i=1}^\infty (x(i+n)+z(i))2^{-i}$ where $n=\lceil\frac{1}{2}\log {\mbox{SNR}}\rceil$ (here we assume a peak power of $1$ for $x$ and $z$). If we think of the transmitted signal $x$ as a sequence of bits at different signal levels, then the ADT model truncates $x$ and passes only its bits above noise level (the first $n$ most significant bits here), i.e., it converts the original Gaussian channel into a deterministic channel without noise.
When applying the ADT model to wireless networks, the broadcasting is captured by the fact that in the resultant deterministic networks, all outgoing edges from the same signal level of any transmitting node carry the same unit information, and the interference is captured by the fact that at each signal level of any receiving node, only the modulo sum of all the received signals is available to the receiving node. This model is called the linear finite-field deterministic channel model in \cite{amir2007_deterministicmodel} \cite{amir2007_wirelessnetworkinfoflow}. We refer to it as the ADT model and denote the finite field of size $p$ associated with the ADT model as $\mathbb{F}_p$ in this paper.

In \cite{amir2007_deterministicmodel} \cite{amir2007_wirelessnetworkinfoflow}, the unicast (i.e., with one source S and one destination D) capacity $C$ of any linear deterministic wireless relay network was characterized as the minimum rank of the adjacency matrices describing all its S-D cuts.
An exhaustive search for finding the minimum rank of the adjacency matrix for all S-D cuts results in an algorithm with complexity exponential in the size of the network.

Amaudruz \& Fragouli \cite{aurore2009_combinatorial_algo_deterministic} were the first to propose a polynomial-time algorithm for finding the unicast capacity of a linear deterministic wireless relay network (see also \cite{fragouli2009_journal}).
In this work, we improve upon Amaudruz \& Fragouli's work and further reduce the computational complexity of the algorithm by fully exploring the useful combinatorial features intrinsic in the problem. Our improvement applies generally with any size of finite fields $\mathbb{F}_p$ associated with the ADT model. Comparing with other algorithms on solving the same problem \cite{sadegh2009_combinatorialstudyofdeterministic} \cite{geomans2009}, our improved algorithm is very competitive in terms of complexity.

This paper is organized as follows. In Section \ref{notation}, we briefly introduce the polynomial-time algorithm by Amaudruz \& Fragouli for finding the unicast capacity of linear deterministic wireless relay networks.  Section \ref{ouralgo} gives a detailed description of our improvement upon the algorithm. First we introduce our improvement with an emphasis on the new components of our algorithm and how they fix the problems within the original algorithm. Then we explore several useful combinatorial features intrinsic in the problem. Finally we explain how these combinatorial features can be combined with our new components to reduce the complexity of the algorithm.
We also give the comparison results between our improved algorithm and other algorithms on solving the same problem.
Section \ref{conclusion} concludes the paper.

\section{Preliminaries and Background\label{notation}}

\subsection{Notations and Definitions}

In \cite{amir2007_wirelessnetworkinfoflow}, it is shown that an arbitrary deterministic relay network can be expanded over time to generate an asymptotically equivalent (in terms of transmission rate) layered network. Therefore, we focus on layered deterministic networks.

Let ${\cal{G=(V,E)}}$ denote a layered deterministic wireless relay network where ${\cal{V}}$ represents the set of nodes in the original wireless relay network, each node in $\mathcal{V}$ has several different levels of inputs and outputs and ${\cal{E}}$ is the set of directed edges going from one input of some node to one output of some other node. For example, Fig. \ref{fig:subfig1} gives a graph representation of a layered deterministic wireless relay network where each node is labeled with a capital letter, all inputs (outputs) from nodes are labeled as $\{x_i\}$ ($\{y_j\}$), $1\leqslant i,j\leqslant 8$.
In the layered network ${\cal{G}}$, all paths from the source node S to the destination node D have equal lengths \cite{amir2007_wirelessnetworkinfoflow}. The set of nodes ${\cal{V}}$ are divided into different layers according to their distances to S. The first layer consists of S and the last layer consists of D.
Let ${\cal{A}}(x_i)$ (or ${\cal{A}}(y_j)$) denote the node where an input $x_i$ (or an output $y_j$) belongs to. Let ${\cal{L}}(A)$ (or ${\cal{L}}(x_i)$, ${\cal{L}}(y_j)$) denote the layer number where node $A$ (or $x_i$, $y_j$) belongs to.
Denote $M$ as the maximum number of nodes in each layer, $L$ the total number of layers and $d$ the maximum number of outgoing edges from any input in any node in the network ${\cal{G}}$ in this paper.

A cut $\Omega$ in ${\cal{G}}$ is a partition of the nodes ${\cal{V}}$ into two disjoint sets $\Omega$ and $\Omega^c$ such that S $\in\Omega$ and D $\in\Omega^c$.
A cut is called a layer cut if all edges across the cut are emanating from nodes from the same layer, otherwise it is called a cross-layer cut. An edge $(x_i,y_j)\in {\cal{E}}$ belongs to layer cut $l$ if ${\cal{L}}(x_i)=l$.

The adjacency matrix $T(\textbf{x},\textbf{y})$ for the sets of inputs $\textbf{x}=\{x_1,x_2,...x_m\}$ and of outputs $\textbf{y}=\{y_1,y_2,...y_n\}$ in ${\cal{G}}$ is a matrix of size $m\times n$ with binary $\{0,1\}$ entries. The rows correspond to $\{x_i\in \textbf{x}\}$ and columns corresponding to $\{y_i\in\textbf{y}\}$ and $T(i,j)=1$ if $(x_i,y_j)\in {\cal{E}}$.
The adjacency matrix $T(E)$ for a set of edges, $E$, is the adjacency matrix for the sets of their inputs and their outputs.

A set of edges, $E$, are said to be linearly independent (LI) if rank$(T(E))=|E|$ (where the rank is computed over GF$(2)$), otherwise they are said to be linearly dependent (LD). In ${\cal{G}}$, each S-D path is of length $L-1$ and crosses each layer cut exactly once. A set of S-D paths are said to be LI if the subsets of their edges crossing each layer cut are LI, otherwise they are said to be LD. In this work, we will consider a slightly more general adjacency matrix, where the non-zero entries can be from a finite field $\mathcal{F}_p$, and the rank is also computed over $\mathcal{F}_p$. Of course, all our results will also apply to the binary field case.

Let ${\cal{E}}_{\Omega}$ be the set of edges crossing the cut $\Omega$ in ${\cal{G}}$. The cut value of $\Omega$ is defined as rank$(T({\cal{E}}_{\Omega}))$, which based on the definition equals the maximum number of LI edges in ${\cal{E}}_{\Omega}$. Note that the cut value defined above is different than that for regular graphs (which is just the number of edges crossing the cut). It is proved \cite{amir2007_deterministicmodel}\cite{amir2007_wirelessnetworkinfoflow} that the unicast capacity of a linear deterministic wireless relay network is equal to the minimum cut value among all S-D cuts.

\subsection{Algorithm by Amaudruz \& Fragouli}\label{originalalgo}

The unicast algorithm by Amaudruz and Fragouli \cite{aurore2009_combinatorial_algo_deterministic} finds the maximum number $C$ of linearly independent S-D paths in a given layered linear deterministic relay network ${\cal{G}}$, where $C$ is the unicast capacity of the network. The algorithm is a path augmentation algorithm, operating in iterations. In each iteration, the algorithm tries to find an additional S-D path so that all S-D paths found are LI. Let ${\cal{P}}=\{{\cal{P}}_{1},...,{\cal{P}}_{k}\}$ denote the set of $k$ LI S-D paths found in the first $k$ iterations. In the process of finding the $(k+1)$-th S-D path ${\cal{P}}_{k+1}$ in iteration $k+1$, the algorithm may make modifications to ${\cal{P}}$ while still maintaining a set of $k$ LI complete S-D paths. The unicast algorithm determines ${\cal{P}}_{k+1}$ by exploring nodes in ${\cal{G}}$ in a certain order as outlined shortly.

The algorithm is implemented in two recursive functions $E_A$ and $E_x$ that explore a node and input respectively. The exploration of a node $A$ takes place when ${\cal{P}}_{k+1}$ has been extended from S to A and needs to be completed from A to D.
In iteration $k+1$, the unicast algorithm calls $E_A$ with the following inputs: ${\cal{G}}$, ${\cal{P}}=\{{\cal{P}}_{1},...,{\cal{P}}_{k}\}$, the indicator function ${\cal{M}}$ (that implements a marking mechanism for visiting nodes and inputs/outputs) and S. The function $E_A$ returns true with one more S-D path ${\cal{P}}_{k+1}$ recorded in ${\cal{P}}$ if it succeeds in finding ${\cal{P}}_{k+1}$, false otherwise.

Exploring node A implies exploring all unused inputs $\{x_i\}$ of A. So we explain the exploration of an input $x_i$ of A below. Hereafter, denote $U^l$ as the sets of used edges by ${\cal{P}}$ in layer cut $l$ and $U_x^l$ and $U_y^l$ as the sets of inputs and outputs used by $U^l$. Let ${\cal{L}}(x_i)=l$. If $x_i\in U_x^l$, do nothing. Otherwise, consider each $y_j$ with $(x_i,y_j)\in{\cal{E}}$ as follows.

\begin{itemize}
\item [(a) ]
{\it $y_j$ is used.} Let $L_{x_i}$ denote the smallest subset of $U_x^l$ with $s=|L_{x_i}|\leqslant|U_x^l|=k$ such that $T(\{L_{x_i},x_i\},U_y^l)$ has rank $s$. The authors prove that replacing any $x_k\in L_{x_i}$ with $x_i$, the algorithm can still maintain $k$ LI S-D paths and the task now is to complete ${\cal{P}}_{k+1}$ from ${\cal{A}}(x_k)$. So in this case the unicast algorithm first finds the set $L_{x_i}$ in function FindL. Then it replaces each $x_k\in L_{x_i}$ with $x_i$ and calls a Match function to find a new set of $k$ edges in layer cut $l$ to maintain $k$ LI S-D paths in ${\cal{P}}$ and tries to complete ${\cal{P}}_{k+1}$ from ${\cal{A}}(x_k)$ if ${\cal{A}}(x_k)$ is not marked or from $x_k$ if $x_k$ is not marked.
We refer to this step as same-layer rewiring.
\item [(b) ]
{\it $y_j$ is not used.} A rank computation function is called on the matrix $T(\{U_x^l,x_i\},\{U_y^l,y_j\})$. If the matrix is not full rank or ${\cal{A}}(y_j)$ has been visited before, do nothing. If the matrix is full rank and ${\cal{A}}(y_j)$ has not been visited before, add $(x_i,y_j)$ to ${\cal{P}}_{k+1}$ and try to complete it from ${\cal{A}}(y_j)$ by exploring ${\cal{A}}(y_j)$. We refer to this step as forward move.

If it fails to complete ${\cal{P}}_{k+1}$ from ${\cal{A}}(y_j)$, a $\phi$-function is called for each $y_k\in U_y^l$ with ${\cal{A}}(y_k)={\cal{A}}(y_j)$. Let ${\cal{P}}_{y_k}$ be the path using $y_k$ and let $(x_k,y_k)\in U^l$ be the path edge. The idea of the $\phi$-function is to complete ${\cal{P}}_{k+1}$ from ${\cal{A}}(y_j)$ to D using the partial path of ${\cal{P}}_{y_k}$ from ${\cal{A}}(y_j)$ to D and then try to complete the path ${\cal{P}}_{y_k}$ from ${\cal{A}}(x_k)$. The $\phi$-function does the following: remove $(x_k,y_k)$ from the set of used edges and try to complete ${\cal{P}}_{y_k}$ from ${\cal{A}}(x_k)$. We refer to this step as backward rewiring.
The $\phi$-function will be executed at most $M$ times.
\end{itemize}

We refer the reader to \cite{aurore2009_combinatorial_algo_deterministic} for more details. The complexity of the algorithm is $O(M\cdot|{\cal{E}}|\cdot C^5)$ and its computational parts include the FindL, Match and rank computation functions each with complexity $O(k^4)$, $O(k^3)$ and $O(k^3)$ respectively.

\subsection{Other Related Algorithms}

Yazdi \& Savari \cite{sadegh2009_combinatorialstudyofdeterministic} developed another polynomial time algorithm with complexity $O(L^8 M^{12} h_0^3+L M^6 C h_0^4)$ (where $h_0$ denotes the maximum total number of inputs/outputs at any layer) by relating matroids with this problem. Most recently, Goemans, Iwata and Zenklusen \cite{geomans2009} proposed a strongly polynomial time algorithm for this problem, whose complexity is $O(L M^3 \log M)$, i.e., it does not depend upon $C$.

\section{Improved Unicast Algorithm\label{ouralgo}}

In this section we outline certain improvements that can be made to the algorithm of \cite{aurore2009_combinatorial_algo_deterministic}. In particular, we elaborate on several useful combinatorial aspects that allow us to reduce the overall time complexity. Moreover, these improvements also fix certain issues with the original algorithm \cite{aurore2009_combinatorial_algo_deterministic}. As mentioned previously, our proposed improvements apply over arbitrary finite fields.

\subsection{Improving the Original Algorithm} \label{improvement}

The main idea in \cite{aurore2009_combinatorial_algo_deterministic} is to find path ${\cal{P}}_{k+1}$ in iteration $k+1$ while maintaining linear independence among all S-D paths in ${\cal{P}}$. In this process, previous paths may be rewired. However, there are cases when the original algorithm may fail to find the exact unicast capacity. We illustrate this using the following examples. We point out that these issues seem to have been resolved in \cite{fragouli2009_journal}. However, our proposed algorithm has several differences from \cite{fragouli2009_journal} as discussed at the end of Section \ref{comparison}.

\noindent
\textbf{Improved Backward Rewiring}

We use the example in Fig. \ref{fig12} to show that there are cases where the $\phi$-function above is insufficient, causing failures of the original algorithm. Then we illustrate how it can be fixed by introducing an improved backward rewiring mechanism.

In Fig. \ref{fig:subfig1}, three LI S-D paths with color red, green and blue are found in the first three iterations of the algorithm. Let's see how the algorithm goes in iteration four. Let's say the algorithm has extended ${\cal{P}}_4$ along the purple path to $y_{20}$. The call $E_A({\cal{G}},{\cal{P}},{\cal{M}},N)$ fails since the only input $x_{24}$ of N is used by paths in ${\cal{P}}$. So $\phi$-function is called on $y_{19}$ and then node I is explored in $E_A({\cal{G}},{\cal{P}},{\cal{M}},I)$, but since there is only one path from all inputs of I to D, $E_A({\cal{G}},{\cal{P}},{\cal{M}},I)$ fails, and finally the algorithm returns false and reports unicast capacity of $3$. However, the unicast capacity of the network is $4$ and a capacity-achieving transmission scheme is given by the four S-D paths in Fig. \ref{fig:subfig2} in different colors.

We propose the following improved backward rewiring mechanism to fix the problem above and to replace the original $\phi$-function.  Let A denote a node in the network (not to be confused with A in the figure). First, the backward rewiring is allowed on every node A whenever it is explored in finding ${\cal{P}}_{k+1}$. Second, the backward rewiring on node A includes the following operations. Let ${\cal{L}}(A)=l+1$.
For any output $y$ of A with $y\in U_y^l$ and $y$ is used by a path in ${\cal{P}}$ at the beginning of the current iteration (if such $y$ exists), (1) find one $x\in U_x^l$ such that $T(U_x^l-x,U_y^l-y)$ has full rank, (2) then rematch $(U_x^l-x,U_y^l-y)$ to generate a new set of $k$ LI used path edges in layer cut $l$ and (3) finally try to complete the partial path from ${\cal{A}}(x)$. Lemma \ref{lemmmm22} guarantees that for a given $y\in U_y^l$ there is always one such $x$ and also a set of edges\footnote{We use the notation $P_{y \rightarrow x}$ since this set of edges can be interpreted as an alternating path, as we show in Section \ref{useful}} $P_{y\rightarrow x}=$ $\{(x_1,y_1=y)$, $(x_1,y_2)$, $(x_2,y_2)$, $(x_2,y_3)$, $...(x_{m'-1},y_{m'})$, $(x_{m'}=x,y_{m'})\}$ $=\{e_1,e_2,...,e_{2m'-1}\}$ with $(x_i,y_i),1\leq i\leq m'$ being edges used by ${\cal{P}}$, which can be found with complexity $O(k^3)$ and $O(k^2)$ respectively. Along the alternating path $P_{y\rightarrow x}$, the rematching of the used path edges in layer cut $l$ can be done easily as follows: $U^l=U^l-e_1+e_2-e_3+...-e_{2m'-1}$.

Consider applying our improved backward rewiring in the example in Fig. \ref{fig12}. It happens on the outputs of nodes N and I. Its application to N is straightforward. Let's look at its application at the output $y_{14}$ of node I. First it finds $x_6\in U_x^2$ with $T(U_x^2-x_6,U_y^2-y_{14})$ having full rank and the alternating path $P_{y_{14}\rightarrow x_6}=\{(x_7,y_{14}),(x_7,y_{13}),(x_6,y_{13})\}$. The rematching is done by $U^2=U^2-(x_7,y_{14})+(x_7,y_{13})-(x_6,y_{13})$. Then node $B={\cal{A}}(x_6)$ is explored. Finally the improved algorithm returns four LI S-D paths in Fig. \ref{fig:subfig2} as expected.

\vspace{-1mm}
\begin{figure}[htbp]
\centering
\subfigure[]{
\includegraphics[scale=0.35]{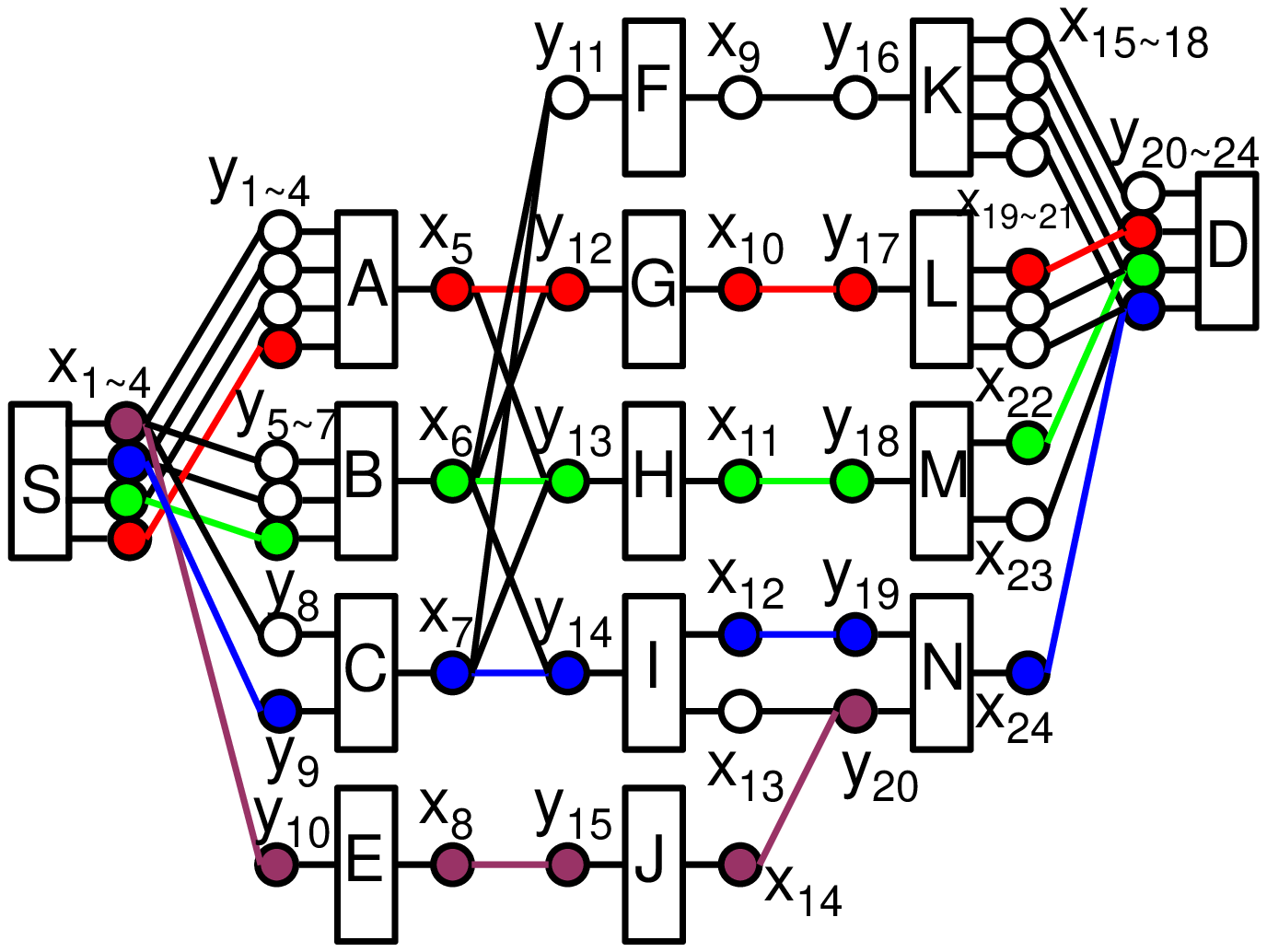}
\label{fig:subfig1}
}
\subfigure[]{
\includegraphics[scale=0.35]{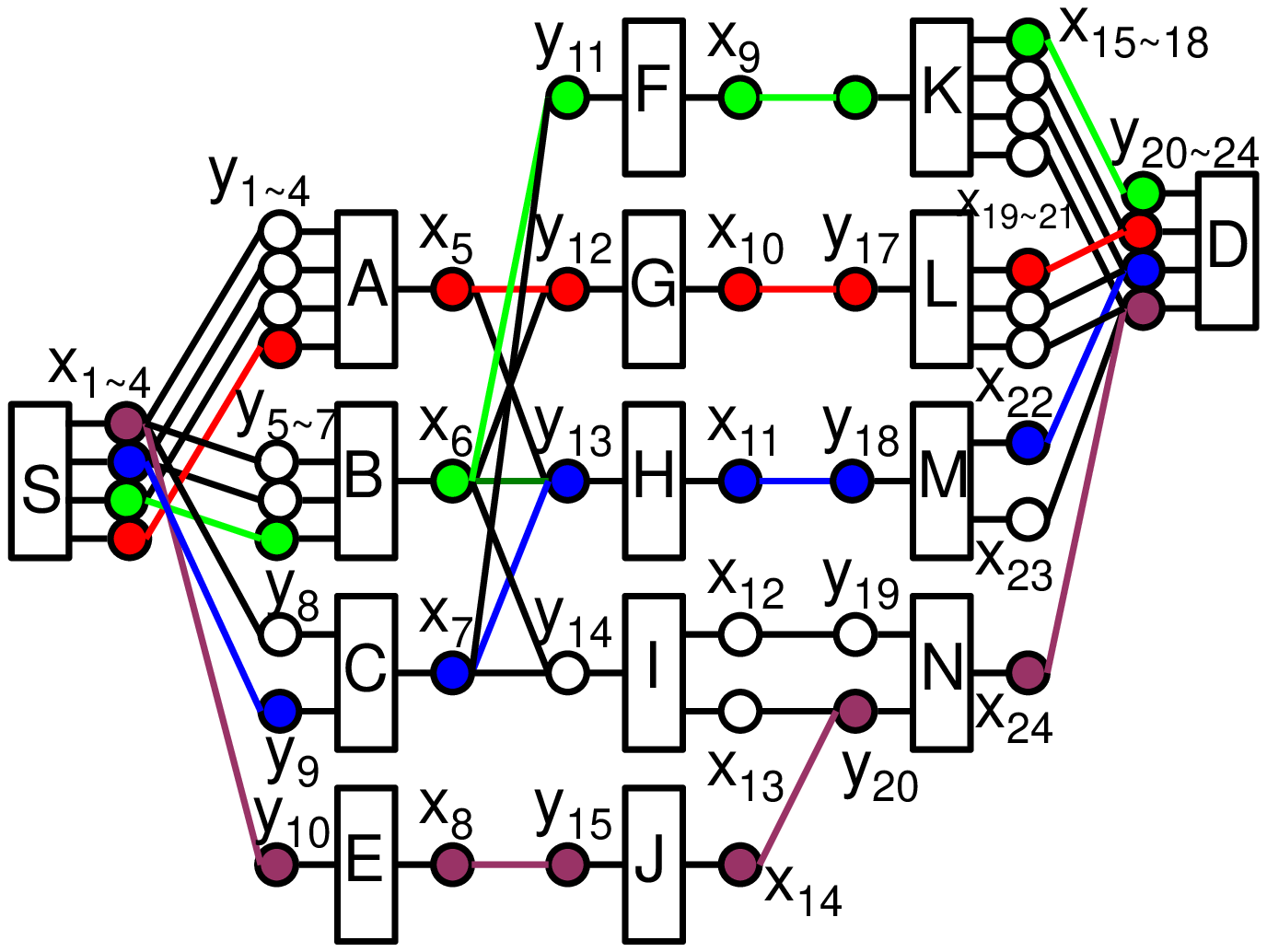}
\label{fig:subfig2}
}
\caption[]{Illustrating example for improved backward rewiring}
\label{fig12}
\end{figure}
\vspace{-1mm}

\noindent
\textbf{Improved Same-Layer Rewiring}

We use the example in Fig. \ref{fig34} to show that the same-layer rewiring in original algorithm is insufficient. Suppose the red S-D path is found in the first iteration. In iteration two, suppose that the algorithm first extends ${\cal{P}}_2$ along the green path to $x_4$. The same-layer rewiring from $x_4$ will mark $x_3$. Since $T(x_3+x_4,y_5+y_6)$ is not full rank, the algorithm fails to complete ${\cal{P}}_2$ along the green path. It continues to extend ${\cal{P}}_2$ along the blue path to $x_5$. Since $x_3$ is marked, the same-layer rewiring from $x_5$ won't be applied on $x_3$ and the call $E_A({\cal{G}},{\cal{P}},{\cal{M}},C)$ fails. The algorithm finally returns false and reports unicast capacity of $1$. However, the network has a unicast capacity of $2$ indicated by the two paths in Fig. \ref{fig:subfig4}.

We develop our improved same-layer rewiring to fix the above problem as follows. First, an input $x_k$ should not be blocked from being visited via same-layer rewiring from any input $x_i$ just because it has been visited via same-layer rewiring from another input $x_j$. Consider the example in Fig. \ref{fig34}. If we allow $x_3$ to be visited via same-layer rewiring from $x_5$, the algorithm may succeed in finding two LI paths as indicated in Fig. \ref{fig:subfig4}. However, this needs to be done carefully. Consider again the example in Fig. \ref{fig34}. If we allow same-layer rewirings from all inputs, then we might run into an infinite loop of going from $x_5$ to $x_3$ via same-layer rewiring and going from $x_3$ to $x_5$ via same-layer rewiring and so on.

The goal of a same-layer rewiring operation in iteration $k+1$ is to ensure that every input, which allows the algorithm to maintain $k$ LI S-D paths and can further extend the current partial path, has the opportunity of being explored, while ensuring that we do not enter an infinite loop. In this work we achieve this by using a pair of labels of each node.

\vspace{-1mm}
\begin{figure}[htbp]
\centering
\subfigure[]{
\includegraphics[scale=0.35]{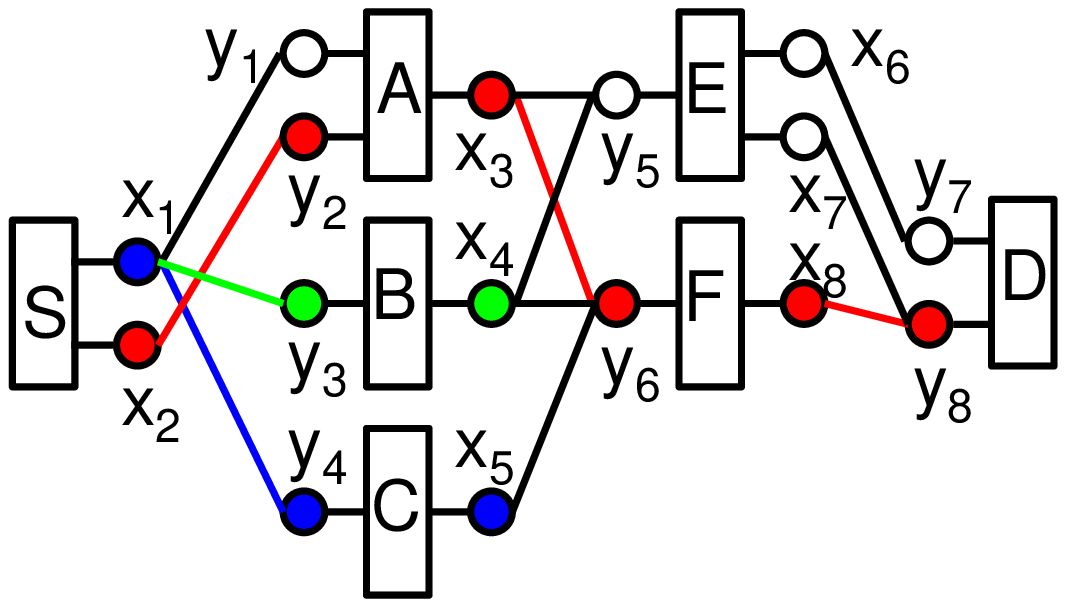}
\label{fig:subfig3}
}
\subfigure[]{
\includegraphics[scale=0.35]{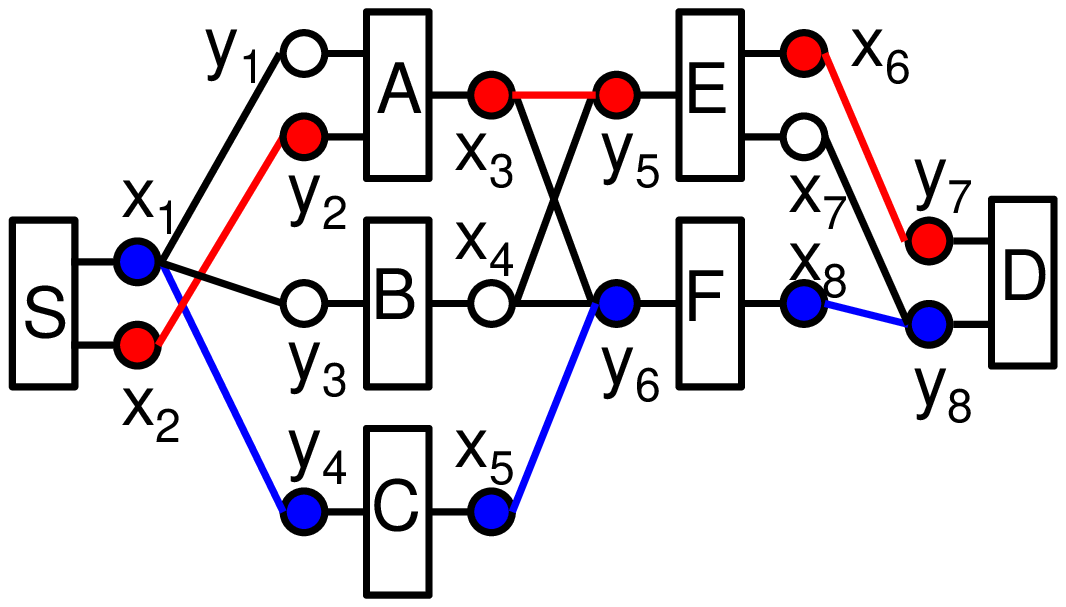}
\label{fig:subfig4}
}
\caption[]{Illustrating example for improved same-layer rewiring}
\label{fig34}
\end{figure}
\vspace{-1mm}

Each node has a label that takes values - ``explored" or ``unexplored". The other label is a type that takes values 1, 2. We initialize the type of every node to be 1 at the beginning of the iteration. A type $1$ input is allowed to initiate same-layer rewirings. An input that is explored via a same-layer rewiring from a type $1$ input $x_i$ is assigned as type $2$. A type $2$ input is not allowed to initiate same-layer rewirings to avoid the possibility of infinite loop. If an input $x$ (of either type) is explored via a backward rewiring, it is re-assigned as type $1$ (since $U_x^l$ and $U_y^l$ change since last time $x$ was explored).

Consider applying our improved same-layer rewiring in the example in Fig. \ref{fig34}. $x_3$ is first visited via a same-layer rewiring from $x_4$ (of type $1$) when it is assigned as type $2$. Later on $x_3$ is revisited via a same-layer rewiring from $x_5$ (of type $1$) when it is assigned as type $2$ again, so it won't initiate a same-layer rewiring to $x_5$, instead it only looks for a possible forward move which happens along the edge $(x_3,y_5)$ (and the improved algorithm finally succeeds in finding $2$ LI paths as in Fig. \ref{fig:subfig4}).

\subsection{Useful Combinatorial Features}\label{useful}

In this subsection, several useful combinatorial features intrinsic in the problem are introduced which are used later in our improved algorithm to reduce the complexity.

In the following, we define a set $\Lambda_{x_i}$ similar to but more general than $L_{x_i}$ in the original algorithm by Amaudruz and Fragouli. $\Lambda_{x_i}$ applies to any size of finite field $\mathbb{F}_p$ associated with the ADT model for the network.

\begin{definition}\label{defin10}
Define $\Lambda_{x_i}$ as a subset of $U_x^{{\cal{L}}(x_i)}$ when $x_i$ is explored such that
\begin{eqnarray}\label{eqn21}
T(x_i,U_y^{{\cal{L}}(x_i)})=\sum_{x_j\in\Lambda_{x_i}}a_{x_i}^j\cdot T(x_j,U_y^{{\cal{L}}(x_i)}).
\end{eqnarray}
where $\{a_{x_i}\}$ are non-zero coefficients from $\mathbb{F}_p$.
\end{definition}

\begin{lemma}\label{lemma00}
$\Lambda_{x_i}$ and the set $\{a_{x_i}\}$ are unique and can be found with complexity $O(k^3)$ in iteration $k+1$.
\end{lemma}

Since $T(U_x^{{\cal{L}}(x_i)},U_y^{{\cal{L}}(x_i)})$ has full-rank, $\Lambda_{x_i}$ and the set $\{a_{x_i}\}$ are unique and can be found with complexity $O(k^3)$ by using Gaussian elimination.

Let ${\cal{G}}_{x_i}$ denote the bipartite graph containing nodes $U_x^{{\cal{L}}(x_i)}\cup U_y^{{\cal{L}}(x_i)}$ when $x_i$ is explored in iteration $k+1$ and ${\cal{G}}_{x_i}^{+}$ denote the bipartite graph containing nodes $\{x_i\}\cup U_x^{{\cal{L}}(x_i)}\cup U_y^{{\cal{L}}(x_i)}$.

In the following, we refer to an alternating path as a path in which the edges belong alternatively to the set of used edges and the set of unused edges.

\begin{lemma}\label{lemma5}
There is an alternating path from $x_i$ to any $x_j\in\Lambda_{x_i}$ in the graph ${\cal{G}}_{x_i}^{+}$ of the form $P_{x_i\rightarrow x_j}=$ $\{(x_i,y_1)$, $(x_1,y_1)$, $(x_1,y_2)$, $(x_2,y_2)$, $...(x_{m-1},y_m)$, $(x_m=x_j,y_m)\}$ with $(x_q,y_q),1\leq q\leq m$ being edges used by ${\cal{P}}$. The complexity for finding these $|\Lambda_{x_i}|$ paths is bounded by $O(k^2)$ in iteration $k+1$.
\end{lemma}
\begin{proof}
Let ${\cal{L}}(x_i)=l$. Given rank$(T(U_x^l,U_y^l))=k$, for any $x_j\in\Lambda_{x_i}$, rank$(T(U_x^l,U_y^l))=$ rank$(T(U_x^l+x_i-x_j,U_y^l))=k$ where $k=|{\cal{P}}|$ in iteration $k+1$. Introduce an auxiliary output $y'$ and an edge $(x_j,y')$. It's easy to see that rank$(T(U_x^l+x_i,U_y^l+y'))=k+1$. Let ${\cal{G}}_{x_i}^{++}$ denote the bipartite graph containing nodes $\{x_i\}\cup U_x^l\cup U_y^l\cup\{y'\}$.

Given $T(U_x^l,U_y^l)$ has full rank, we know that the polynomial of the determinant of the Edmonds matrix of the bipartite graph ${\cal{G}}_{x_i}$ is not identically zero, so there is a size $k$ perfect matching in ${\cal{G}}_{x_i}$ \cite{rajeev1995_randomizedalgorithm}, $M_1=U^l$ giving such a matching. Similarly given rank$(T(U_x^l+x_i,U_y^l+y'))=k+1$, there is a size $k+1$ perfect matching in ${\cal{G}}_{x_i}^{++}$. By Berge's Lemma \cite{berge1957}, we know that there is an alternating path, relative to the matching $M_1$, starting from an unused input $x_i$ to an unused output $y'$, alternating between edges not in the current matching $M_1$ and edges in the current matching $M_1$, i.e., there is a path $P_{x_i\rightarrow y'}=$ $\{(x_i,y_1)$, $(x_1,y_1)$, $(x_1,y_2)$, $(x_2,y_2)$, $...(x_{m-1},y_m)$, $(x_m,y_m)$, $(x_m=x_j,y')\}$ with $(x_q,y_q),1\leq q\leq m$ being edges in $M_1$. So we proved that there is an alternating path $P_{x_i\rightarrow x_j}=$ $\{(x_i,y_1)$, $(x_1,y_1)$, $(x_1,y_2)$, $(x_2,y_2)$, $...(x_{m-1},y_m)$, $(x_m=x_j,y_m)\}$ with $(x_q,y_q),1\leq q\leq m$ being edges in $M_1=U^l$.

Since the number of nodes in ${\cal{G}}_{x_i}^{+}$ is bounded by $O(k)$, the number of its edges is bounded by $O(k^2)$. Finding $P_{x_i\rightarrow x_j}$ for all $x_j\in\Lambda_{x_i}$ in ${\cal{G}}_{x_i}^{+}$ can be done with complexity $O(k^2)$ with some well-known graph traversal algorithms, like breadth-first search \cite{CRLS_2001}.
\end{proof}

\begin{lemma}\label{lemmmm22}
Let rank$(T(U_x^l,U_y^l))=|U_x^l|=|U_y^l|=k+1$. Given any $y\in U_y^l$, there exists at least one $x\in U_x^l$, such that rank$(T(U_x^l-x,U_y^l-y))=k$. Moreover there is an alternating path from $y$ to $x$ of the form $P_{y\rightarrow x}=$ $\{(x_1,y_1=y)$, $(x_1,y_2)$, $(x_2,y_2)$, $(x_2,y_3)$, $...(x_{m'-1},y_{m'})$, $(x_{m'}=x,y_{m'})\}$ with $(x_q,y_q),1\leq q\leq m'$ being edges in $U^l$. The complexity of finding one such $x$ is bounded by $O(k^3)$ and the complexity of finding path $P_{y\rightarrow x}$ is bounded by $O(k^2)$.
\end{lemma}

Due to lack of space, we skip the proof here. The proof of existence of $P_{y\rightarrow x}$ is similar to Lemma \ref{lemma5} by introducing an auxiliary input $x'$ and output $y'$ and edges $(x',y),(x,y')$ leading to rank$(T(U_x^l+x',U_y^l+y'))=k+2$.

Lemma \ref{lemma2} develops an equivalent but computationally simple method to speed up the rank computation when $x_i$ is explored given $\Lambda_{x_i}$ and the set of associated coefficients $\{a_{x_i}\}$.

\begin{lemma}\label{lemma2}
Let $T(U_x^l,U_y^l)$ have full rank $k$. The rank computation for checking rank$(T(U_x^l+x_i,U_y^l+y))=k$ or $k+1$ for any $x_i\not\in U_x^l$, ${\cal{L}}(x_i)=l$, $y\not\in U_y^l$ and $(x_i,y)\in{\cal{E}}$ is equivalent to checking $T(x_i,y)=\sum_{x_j\in\Lambda_{x_i}}a_{x_i}^j\cdot T(x_j,y)$ or not, with complexity bounded by $O(k)$ given $\Lambda_{x_i}$ and $\{a_{x_i}\}$.
\end{lemma}
\begin{proof}
Given $T(U_x^l,U_y^l)$ has full rank $k$, rank$(T(U_x^l+x_i,U_y^l+y))=k$ is equivalent to that
$T(x_i,U_y^l+y)=\sum_{x_j\in\Lambda_{x_i}'}a_{x_i'}^{j}\cdot T(x_j,U_y^l+y)$ for some $\Lambda_{x_i}'\subseteq U_x^l$ and $\{a_{x_i'}\}$. Since $\Lambda_{x_i}\subseteq U_x^l$ and the set $\{a_{x_i}\}$ are unique for which $T(x_i,U_y^l)=\sum_{x_j\in\Lambda_{x_i}}a_{x_i}^j\cdot T(x_j,U_y^l)$ holds (by Lemma \ref{lemma00}), there must be $\Lambda_{x_i}'=\Lambda_{x_i}$ and $\{a_{x_i}\}=\{a_{x_i'}\}$. This leads to that rank$(T(U_x^l+x_i,U_y^l+y))=k$ is equivalent to $T(x_i,y)=\sum_{x_j\in\Lambda_{x_i}}a_{x_i}^j\cdot T(x_j,y)$.
\end{proof}

\begin{lemma} \label{lemmaappen16}
Let $x'\in\Lambda_{x_i}$. If $x'$ is explored via a same-layer rewiring from $x_i$, $\Lambda_{x'}=\Lambda_{x_i}+x_i-x'$ and the set of associated coefficients $\{a_{x'}\}$ can be computed from $\{a_{x_i}\}$ with complexity $O(k)$ in iteration $k+1$.
\end{lemma}
\begin{proof}
Let ${\cal{L}}(x_i)=l$. Note that when $x'$ is explored via a same-layer rewiring from $x_i$, $U_x^l$ is updated as $U_x^l-x'+x_i$, $U_y^l$ is unchanged and $T(U_x^l-x'+x_i,U_y^l)$ has full rank. Based on definition,
\begin{eqnarray}\label{eqn81}
T(x_i,U_y^l)=\sum_{x_j\in\Lambda_{x_i}\char92 x'}a_{x_i}^j\cdot T(x_j,U_y^l)+a_{x_i}'\cdot T(x',U_y^l).
\end{eqnarray}
where $\{a_{x_i}\}$ are non-zero coefficients from $\mathbb{F}_p$. So we have
\begin{eqnarray}\label{eqn82}
T(x',U_y^l)=\sum_{x_j\in\Lambda_{x_i}\char92 x'}\frac{a_{x_i}^j}{a_{x_i}'}\cdot T(x_j,U_y^l)-\frac{1}{a_{x_i}'}\cdot T(x_i,U_y^l).
\end{eqnarray}
Since $T(U_x^l-x'+x_i,U_y^l)$ has full rank, equation (\ref{eqn82}) is the unique way that the row $T(x',U_y^l)$ can be expressed as a linear combination of the rows in this matrix.
So we conclude $\Lambda_{x'}=\Lambda_{x_i}+x_i-x'$ and the set of associated coefficients $\{a_{x'}\}$ can be computed from $\{a_{x_i}\}$ with complexity $O(k)$. Note that in iteration $k+1$, $|\Lambda_{x_i}|\leqslant|U_x^l|=k$.
\end{proof}

\subsection{Reducing the Complexity and the Overall Algorithm}

As mentioned before, the computational parts of algorithm \cite{aurore2009_combinatorial_algo_deterministic} include the FindL (finding $L_{x_i}$), Match (update $U$ after a same-layer rewiring from $x_i$) and rank computation functions. Now we explain how the combinatorial features from Section \ref{useful} can be used to further reduce the complexity of the unicast algorithm.

Lemma \ref{lemma00} shows that $\Lambda_{x_i}$ and the set of associated coefficients $\{a_{x_i}\}$ for any type $1$ input $x_i$ can be computed with complexity $O(k^3)$ in iteration $k+1$. Lemma \ref{lemmaappen16} tells that for any type $2$ input $x'$, $x'\in\Lambda_{x_i}$, that is explored via a same-layer rewiring from a type $1$ input $x_i$, $\Lambda_{x'}$ and the set of associated coefficients $\{a_{x'}\}$ can be computed with complexity $O(k)$ given $\Lambda_{x_i}$ and the set of associated coefficients $\{a_{x_i}\}$.

Second, based on Lemma \ref{lemma5}, the matching or updating of $U$ after same-layer rewirings from any type $1$ input $x_i$ can be done with complexity $O(k^2)$ in iteration $k+1$ as follows. First find all $|\Lambda_{x_i}|$ paths $P_{x_i\rightarrow x_j}$, $\forall x_j\in\Lambda_{x_i}$ with complexity $O(k^2)$ for $x_i$. Let $P_{x_i\rightarrow x_j}=$ $\{(x_i,y_1)$, $(x_1,y_1)$, $(x_1,y_2)$, $...(x_{m-1},y_m)$, $(x_m=x_j,y_m)\}$ $=\{e_1,e_2,...,e_{2m}\}$ with $(x_q,y_q),1\leq q\leq m$ being edges used by ${\cal{P}}$ for any $x_j\in\Lambda_{x_i}$. Then updating of $U^{{\cal{L}}(x_i)}$ after a same-layer rewiring from $x_i$ to $x_j$ can be done by $U^{{\cal{L}}(x_i)}\leftarrow U^{{\cal{L}}(x_i)}+e_1-e_2+...-e_{2m}$.

Third, Lemma \ref{lemma2} tells that the rank computation in a forward move from any $x_i$ (either of type $1$ or of type $2$), $x_i\not\in U_x^l$, ${\cal{L}}(x_i)=l$, for checking rank$(T(U_x^l+x_i,U_y^l+y))=k$ or $k+1$ for any $y\not\in U_y^l$ and $(x_i,y)\in{\cal{E}}$ is equivalent to checking $T(x_i,y)=\sum_{x_j\in\Lambda_{x_i}}a_{x_i}^j\cdot T(x_j,y)$ or not, with complexity bounded by $O(k)$ given $\Lambda_{x_i}$ and $\{a_{x_i}\}$ in iteration $k+1$.

Finally, as mentioned before, in our improved backward rewiring from an output $y$, to find one $x$ with $T(U_x^l-x,U_y^l-y)$ having full rank and to rematch $(U_x^l-x,U_y^l-y)$ can be done with complexity $O(k^3)$ in iteration $k+1$ guaranteed by Lemma \ref{lemmmm22}.

Table \ref{table1} gives an overall description of our improved unicast algorithm which is implemented in a function $E_A({\cal{G}},{\cal{P}},{\cal{M}},A)$ where all inputs are the same as in the original algorithm. A complete software implementation of our improved unicast algorithm can be found in \cite{myhomepage}.

\begin{table}[htbp]
\caption{Pseudo-code for our improved algorithm}
\vspace{-4mm}
\begin{center}
\resizebox{1.0\columnwidth}{!}{
\begin{tabular}{l}
\hline
\hline
\{(T,F)\}=$E_A({\cal{G}},{\cal{P}},{\cal{M}},A)$\\
$\left\{\begin{array}{l}
{\cal{M}}(A)=T, {\cal{L}}(A)=l\\
U^l=\{{\mbox{used edges in layer cut }}l\}, U_x^l=\{x_i\in U^l\}, U_y^l=\{y_j\in U^l\}\\
{\mbox{for any }}x: {\cal{A}}(x)=A, x\not\in U_x^l, {\cal{M}}(x)=F, {\mbox{GetType}}(x)=2\\
\left\{\begin{array}{l}
{\cal{M}}(x)=T\\
{\mbox{for any }}y: (x,y)\in {\cal{E}}, y\not\in U_y^l, {\cal{M}}({\cal{A}}(y))=F {\mbox{ //\emph{forward move}}}\\
\left\{\begin{array}{l}
{\mbox{if }}T(x,y)\neq\sum_{x_j\in\Lambda_{x}}a_x^j\cdot T(x_j,y)\\
\left\{\begin{array}{l}
{\mbox{Update}}({\cal{P}}); U^l\leftarrow U^l+e\\
{\mbox{if }}{\cal{A}}(y)=D, {\mbox{ return (T)}}\\
{\mbox{else if }}E_A({\cal{G}},{\cal{P}},{\cal{M}},{\cal{A}}(y))=T, {\mbox{ return(T)}}\\
U^l\leftarrow U^l-e; {\mbox{ Restore}}({\cal{P}})\\
\end{array}\right.\\
\end{array}\right.\\
\end{array}\right.\\
{\mbox{for any }}x: {\cal{A}}(x)=A, x\not\in U_x^l, {\cal{M}}(x)=F, {\mbox{GetType}}(x)=1\\
\left\{\begin{array}{l}
{\cal{M}}(x)=T\\
{\mbox{Compute }}\Lambda_{x}{\mbox{ and the set of coefficients }}\{a_x\}\\
{\mbox{for any }}y: (x,y)\in {\cal{E}}, y\not\in U_y^l, {\cal{M}}({\cal{A}}(y))=F{\mbox{ //\emph{forward move}}}\\
\left\{\begin{array}{l}
{\mbox{if }}T(x,y)\neq\sum_{x_j\in\Lambda_{x}}a_x^j\cdot T(x_j,y)\\
\left\{\begin{array}{l}
{\mbox{Update}}({\cal{P}}); U^l\leftarrow U^l+e \\
{\mbox{if }}{\cal{A}}(y)=D, {\mbox{ return (T)}}\\
{\mbox{else if }}E_A({\cal{G}},{\cal{P}},{\cal{M}},{\cal{A}}(y))=T, {\mbox{ return(T)}}\\
U^l\leftarrow U^l-e; {\mbox{ Restore}}({\cal{P}})\\
\end{array}\right.\\
\end{array}\right.\\
{\mbox{Find all paths }}P_{x\rightarrow x_j}{\mbox{ for all }}\forall x_j\in\Lambda_{x}\\
{\mbox{for any }}x_j: x_j\in\Lambda_{x}{\mbox{ with }}P_{x\rightarrow x_j}=\{e_1,e_2,...e_{2m}\}=\\
\{(x,y_1),(x_1,y_1),(x_1,y_2),...(x_m=x_j,y_m)\}{\mbox{ //\emph{same-layer rewiring}}}\\
\left\{\begin{array}{l}
{\cal{M}}(x_j)=F; {\mbox{ SetType}}(x_j,2);\\
\Lambda_{x_j}=\Lambda_{x}-x_j+x\\
{\mbox{compute }}\{a_{x_j}\}{\mbox{ based on }}\{a_x\}{\mbox{ according to Lemma \ref{lemmaappen16}}}\\
{\mbox{Update}}({\cal{P}}); U^l\leftarrow U^l+e_1-e_2+...+e_{2m-1}-e_{2m}\\
{\mbox{if }}E_A({\cal{G}},{\cal{P}},{\cal{M}},{\cal{A}}(x_j))=T, {\mbox{ return(T)}}\\
U^l\leftarrow U^l-e_1+e_2-...-e_{2m-1}+e_{2m}; {\mbox{ Restore}}({\cal{P}})\\
\end{array}\right.\\
\end{array}\right.\\
{\mbox{for any }}y: {\cal{A}}(y)=A, y\in U_y^{l-1}, {\cal{M}}(y)=F \\
{\mbox{and }}y {\mbox{ is used by }}{\cal{P}} {\mbox{ at the beginning of the iteration}}{\mbox{ //\emph{backward rewiring}}}\\
\left\{\begin{array}{l}
{\cal{M}}(y)=T\\
{\mbox{find one }}x\in U_x^{l-1}{\mbox{ with }}T(U_x^{l-1}-x,U_y^{l-1}-y){\mbox{ having full rank}}\\
{\mbox{and find }}P_{y\rightarrow x}=\{e_1,e_2,...e_{2m'-1}\}\\
=\{(x_1,y_1=y),(x_1,y_2),(x_2,y_2),...(x_{m'}=x,y_{m'})\}\\
{\cal{M}}(x)=F, {\mbox{SetType}}(x,1)\\
{\mbox{Update}}({\cal{P}}); U^{l-1}\leftarrow U^{l-1}-e_1+e_2-...-e_{2m'-1}\\
{\mbox{If }}E_A({\cal{G}},{\cal{P}},{\cal{M}},{\cal{A}}(x))=T,{\mbox{ return (T)}}\\
U^{l-1}\leftarrow U^{l-1}+e_1-e_2+...+e_{2m'-1};{\mbox{ Restore}}({\cal{P}})\\
\end{array}\right.\\
{\mbox{return (F)}}\\
\end{array}\right.$\\
\hline
\hline
\end{tabular}
}
\end{center}
\label{table1}
\end{table}

\vspace{-3mm} 

\subsection{Complexity Analysis and Comparison with Existing Results}\label{comparison}

To analyze the complexity, we first bound the total number of inputs of different types being visited in each iteration of the algorithm. Note that once a node or input/output is visited/explored, it's labeled as explored (by ${\cal{M}}$) and not allowed to be explored again unless it is relabeled as unexplored again. At the beginning of each iteration, all inputs are initialized as unexplored type $1$ inputs whose number is bounded by $O(|{\cal{V}}_x|)$ (let ${\cal{V}}_x=$\{all inputs in the network\}). In each backward rewiring operation, one input will be assigned as unexplored type $1$ input. From the definition of backward rewiring, the total number of valid outputs that initiate a backward rewiring is no more than $|{\cal{V}}_x|$, which means the total number of backward rewiring operations is bounded by $O(|{\cal{V}}_x|)$. So the total number of type $1$ inputs being visited is bounded by $O(|{\cal{V}}_x|)$ in each iteration. In each same-layer rewiring operation from a type $1$ input, one input will be assigned as unexplored type $2$ input. The total number of same-layer rewiring operations from any type $1$ input $x$ is no more than $|\Lambda_{x}|\leqslant k$ in iteration $k+1$. So the total number of type $2$ inputs being visited is bounded by $O(k|{\cal{V}}_x|)$ in iteration $k+1$.

The worst case in computation in iteration $k+1$ are no more than: (1) for each type $1$ input $x_i$, compute $\Lambda_{x_i}$ and $\{a_{x_i}\}$ with complexity $O(k^3)$ and find all paths $P_{x_i\rightarrow x_j}$ for $\forall x_j\in\Lambda_{x_i}$ with complexity $O(k^2)$, (2) for each type $2$ input $x_j$, compute $\Lambda_{x_j}$ and $\{a_{x_j}\}$ with complexity $O(k)$, (3) for each type $1$ or type $2$ input $x$, compute rank for $T(U_x^l+x,U_y^l+y)$ for all $y\not\in U_y^l$, $(x,y)\in{\cal{E}}$ with complexity $O(k)$ given $\Lambda_x$ and $\{a_x\}$ (for any $x$, the total number of such $y$ is no larger than $d$) and (4) in each backward rewiring from a certain $y$, find one $x$ with $T(U_x^l-x,U_y^l-y)$ having full rank and to rematch $(U_x^l-x,U_y^l-y)$ with complexity $O(k^3)$. Note that $k\leqslant C$. It's obvious that the total complexity of our improved algorithm is bounded by $O(|{\cal{V}}_x|\cdot C^4+d\cdot|{\cal{V}}_x|\cdot C^3)$.

Due to lack of space, we skip the proof of correctness for our improved algorithm, however a complete and detailed proof can be found in \cite{myhomepage}.

Table \ref{table2} lists the comparison results between different algorithms for finding the unicast capacity of linear deterministic wireless relay networks, specially in their complexity.

\begin{table}[htbp]
\caption{Comparison of algorithm complexity}
\vspace{-4mm}
\begin{center}
\resizebox{1.0\columnwidth}{!}{
\begin{threeparttable}
\begin{tabular}{l|p{32mm}|p{58mm}}
\hline
\hline
Algorithm&Complexity{\tnote{$\ast$}}&Notes\\
\hline
\cite{aurore2009_combinatorial_algo_deterministic}&$O(M |{\cal{E}}| C^5)$&Always higher than ours\\
\cite{fragouli2009_journal}&$O(d|{\cal{V}}_x|C^5+|{\cal{V}}_y|C^5)$&especially when $C$ is large\\
\hline
\cite{sadegh2009_combinatorialstudyofdeterministic}&$O(L^8 M^{12} h_0^3+L M^6 C h_0^4)$&Always higher than ours, especially when $M$ or $L$ is large\\
\hline
\cite{geomans2009}&$O(L^{1.5} M^{3.5} \log(ML))$ or $O(L M^3 \log M)$&Straightforward comparison is not possible. \cite{geomans2009} will have lower complexity if $C$ is much larger than $M$\\
\hline
Our work&$O(|{\cal{V}}_x| C^4+d |{\cal{V}}_x| C^3)$&-\\
\hline
\hline
\end{tabular}
\begin{tablenotes}
\item[$\ast$]{Denote $C$ as the unicast capacity, $M$ the maximum number of nodes in each layer, $L$ the total number of layers, $d$ the maximum number of inputs of any node, $h_0$ the maximum number of inputs/outputs at any layer, $E$ the total number of edges, $|{\cal{V}}_x|$ the total number of inputs and $|{\cal{V}}_y|$ the total number of outputs. Note that $M\geq d$ (since by definition each input can have at most one connection to each node in the next layer), $|{\cal{E}}|\geq|{\cal{V}}_x|$ (because of broadcasting) and $h_0\geq C$ (based on definition).}
\end{tablenotes}
\end{threeparttable}
}
\end{center}
\label{table2}
\end{table}

\vspace{-3mm} 

We note that the issues with the original algorithm \cite{aurore2009_combinatorial_algo_deterministic} mentioned in Section \ref{improvement} have been fixed in \cite{fragouli2009_journal}. The main difference between our improved algorithm and the algorithm in \cite{fragouli2009_journal} is that our improved algorithm utilizes those useful combinatorial features intrinsic in the problem described in Section \ref{useful} which lead to reduced complexity. The other difference comes from the same-layer rewiring and backward rewiring. In \cite{fragouli2009_journal}, the same-layer rewiring starts on each input at most once (using the ML indicator function) while our algorithm allows multiple same-layer rewirings starting from certain inputs (that is, if an input is explored via a backward rewiring, it is reassigned as type $1$ input and allows to initiate same-layer rewiring again). In \cite{fragouli2009_journal}, the backward rewiring (implemented in $\phi$-function there) allows exploration on every $x_k\in U_x$ such that the resulting adjacency matrix of used path edges still remains full rank while our algorithm only finds one such $x_k\in U_x$ and explores it. Note that it can be verified that the combined effects of the different same-layer rewiring and backward rewiring in two algorithms are the same.

\section{Conclusions\label{conclusion}}

An improved algorithm for finding the unicast capacity of linear deterministic wireless networks is presented. Our algorithm improves upon the original algorithm by Amaudruz \& Fragouli. We amend the original algorithm so that it finds the unicast capacity correctly for any given deterministic networks. Moreover we fully explore several useful combinatorial features intrinsic in the problem which lead to reduced complexity. Our improved algorithm applies with any size of finite field associated with the ADT model defining the network. Our improved algorithm proves to be very competitive when comparing with other algorithms on solving the same problem in terms of complexity.

\bibliographystyle{IEEEtran}
\bibliography{main}

\end{document}